\documentclass[twocolumn,nofootinbib,amsmath]{revtex4}  % for review and submission
\usepackage{graphicx}  % needed for figures
\usepackage{dcolumn}   % needed for some tables
\usepackage{bm}        % for math
\usepackage{amssymb}   % for math
\usepackage{makecell} 

\usepackage{amsthm}
\newtheorem{thm}{Theorem}
\begin{document}
%{\scriptsize
%\textbf{
%\vfill\noindent To the memory of my teacher \& father,\\
% Hossein Porforough \\
%}
%}

\title{Rediscovery of Cooper pair factor}
\author{M. Porforough}
\email{milaad@aut.ac.ir}

\affiliation{Department of Physics, School of Sciences, Tarbiat Modares University, P.O.Box 14155-4838, Tehran, Iran}

\date{\today}

\begin{abstract}
In superconductivity, the Cooper pair \cite{Ref1}  is a bound state of a pair of weakly interacting electrons in a metal
 which forces us to insert a factor of $2e$ (instead of $e$) in the phase of the wave function in the Aharonov-Bohm (AB) effect. We claim that the existence of such a factor is universal. In \cite{Ref2}, we proposed an explanation for the AB effect (or equivalently, a supersymmetric Dirac monopole) based on SUSY. It is shown that this model opens a path to recapture this factor without the need for arguments related to attractive pairs. The idea  is also applied to the case of a nonlinear sigma model (NLSM) when it is coupled to  $\mathcal{N}=1$, supergravity.
\end{abstract}

\maketitle

\section{Introduction}

A bound state of the Cooper pair  \cite{Ref1} is an elementary object which causes superconductivity. Cooper showed that the attraction of the pair
is originated by an electron-phonon interaction which causes the energy of the pair to lie lower than the Fermi energy. One year later, Bardeen, Cooper and Schrieffer developed the theory of the interacting Cooper pairs to describe the superconductivity which is known as the BCS theory \cite{Ref3}.

From the theoretical perspective, one can ask: can we prove the existence of the Cooper pair which implies that there is something more profound out there? Is the Cooper factor a universal notion or it is just restricted to the superconductivity context? We will try to find an answer to these types of questions  in the following sections. 

The wave function of an electron which moves along a closed path in the presence of a gauge connection $\mathbf{A}(\mathbf{r})$ is accompanied by a phase factor:
\begin{align}\label{1}
\psi(\mathbf{r})=\psi_0 (\mathbf{r}) \; \mathrm{e}^{-i(e  /\hbar) \oint \mathrm{d}\mathbf{r}.\mathbf{A}(\mathbf{r})},
\end{align}
where $\psi_0(\mathbf{r})$ denotes the wave function at  point $\mathbf{r}$ at the beginning and $\psi (\mathbf{r}) $ denotes the wave function at the same point after traveling around the closed path. The integral which appears in the exponent, represents the magnetic flux $\Phi_m$ enclosed by the loop. Note that in this paper, we will work in the SI units with the Weber convention. The wave function at $\mathbf{r}$ should be single-valued so we obtain the magnetic flux quantization:
\begin{align}\label{2}
\Phi_m=\frac{2\pi \hbar}{e}n \;\;\;\; n=0,\pm 1, \pm 2, \cdots .
\end{align}
Since a superconductor expels magnetic field lines, it is a unique material by which one can measure this flux and verify the relation  (\ref{2}). Experimental results show that an extra factor of $2$ appears in the denominator of the relation  (\ref{2}). Due to the Cooper pair, this is consistent with our theoretical understanding of superconductivity.

In the next section, we will consider a supersymmetric Dirac monopole and show that  in this model, the Cooper factor arises  automatically. In the third section, we will verify this prescription in the case of a NLSM coupled to the $d=4$, $\mathcal{N}=1$ supergravity. The paper closes in the fourth section, with a brief concluding remark. \\

\section{A supersymmetric Dirac monopole}

In this section, first, we will show that a supersymmetric Dirac monopole admits a geometrical phase which can be expressed in terms of  FI term, and after that, we will derive the Cooper factor which is nicely in agreement with the statement  that we made below the equation (\ref{2}). It is an intersecting D3-D3' system in which we can put some charged matters in the intersection region. The magnetic property of this set-up  was discovered by Mintun, Polchinski and Sun (MPS) in 2015 \cite{Ref4}. For our purposes, it is not necessary to use the  brane interpretation of the system because we just want to focus carefully on its gauge properties. The world-volume theory of each brane is a usual 4d super Yang-Mills theory with a compact gauge group U(1), each of which contains a region in which the charges are localized;  namely, the intersection region. We will develop the formulation from the gauge theory point of view. However, to illuminate some features like the number of supercharges or chiral multiplets, we will mention the constraints that arise from the string theory to make the discussion clearer.

We want to construct a supersymmetric Dirac monopole. The Dirac monopole is a point-like particle in a specific position in space e.g. the origin. Knowing the position of the monopole is a crucial fact to analyze its behavior because, as you may know, it is a topological object—i.e. the Maxwell equation $\nabla . \mathbf{B}=0$ holds everywhere in the space except at the position of the monopole. Therefore, we should remove this point from the space (the $\mathbb{R}^3$ manifold). The remaining manifold is homotopy equivalent to a $\mathbf{S}^2$ which surrounds the origin. To understand the deep topology behind  the AB effect and the Dirac monopole, see the famous paper \cite{Ref5}. 

So far, we have seen the importance of the position of the monopole. Now, in a 4d Lorentzian manifold, we want to put some U(1) charges such that the configuration preserves some amounts of supersymmetry and besides that, the location of charges should be seen directly. The simplest submanifold corresponding to the SO(1,3) algebra to put charges on it in this way, corresponds to the SO(1,1) algebra. So inside the 4d space with coordinates $x^\mu; \mu=0,1,2,3$, there is a special 2d Lorentzian submanifold with coordinates $x^\mu; \mu=0,1$, on which we can localize a matter content (belonging to a chiral multiplet), and then charge it by coupling it to the 4d gauge field (belonging to a vector multiplet), in a supersymmetric way. Adding this 2d action to the 4d sourceless gauge theory gives the full theory. At this time, we still do not know if the configuration contains BPS solutions for the magnetic monopole or not. Note that the supersymmetry imposes a strong constraint that if the Dirac monopole exists, it is no longer a point-like particle in the space; it is a line or, more literally, a line defect. This is because only a line can cause a breaking of the initial SO(1,3) to that SO(1,1). For this SO(1,1), it is possible to extend the corresponding Poincaré algebra to a super-Poincaré algebra which is also a 
${\displaystyle \mathbb {Z} _{2}}$-graded Lie superalgebra. This algebra is a non-trivial extension of the original Poincaré algebra that avoids the Coleman-Mandula’s no-go theorem. However, subject to supercharges $Q_{\pm},\bar{Q}_{\pm}$ in two dimensions, we have new operators that do not correspond to the generators of the (internal) gauge groups which should be related by commutators rather than anti-commutators. For example, recall the anti-commutator $ \left \{ Q_+,Q_- \right \}=Z^*$, where $Z$ is the complex central charge.

As explained  in \cite{Ref4}, firstly, the world-volume theory of each brane contains  magnetic solutions of the BPS equations when it is coupled to a NLSM. Secondly, one can construct the full theory by using some beautiful techniques. We do not want to get involved  with this procedure in detail because one can make a subtle guess to predict the general form of the most important sector of the final theory. For those who are interested in technical problems, we would say that the basis of the procedure is as follows: one can start with suitable $\mathcal{N}=1, d=4$ multiplets on the intersection region and use the bottom-top approach \cite{Ref6} to make a supersymmetric theory in the higher dimensions. So putting charges in a subregion is totally under control. If we consider this system as a tiny part of the string theory, then the number of supercharges is forced to be 8; but even 4 is sufficient for our purposes. In general, we can build up such a configuration with any amount of supercharges. The building blocks of this process,  as we mentioned above, are $\mathcal{N}=1, d=4$ multiplets. As long as the amount of supercharge increases, the construction becomes more and more difficult because there appear to be a large number of constraints such as the $\mathcal{R}$ symmetry which gets us into massive trouble.

We are interested only in the magnetic part of the Maxwell sector of the full theory. Fortunately, one can guess that the equation of motion of this sector is a usual 4d Maxwell equation in the presence of a magnetic source which is localized in regions on which we have $x_2=x_3=0$, by a suitable delta function. We express this equation by the language of differential forms:
\begin{align}\label{3}
\mathrm{d}\mathcal{F}=\frac{1}{c}*j_m.
\end{align}
Using the bottom-top approach, we can immediately conclude that the magnetic source arises from the D-term of the full theory. The delta function on the RHS of (\ref{3}) implies that at any instant of time, in every point outside of the line we have: $\mathrm{d}\mathcal{F}=0$. 

Our main goal is to extract the geometrical phase, so first we should investigate Poincar\'e's lemma. Outside of the line, $\mathcal{F}$ is closed. Can we conclude that $\mathcal{F}$ is exact in this region too? According to Poincar\'e's lemma, if $U$ denotes a coordinate neighbourhood of a manifold $M$, our answer is yes if and only if $U$ is shrinkable to a point $p_0$ in $M$. First, consider the usual point-like Dirac monopole placed at the origin of $\mathbb{R}^3$. For now, we do not need to worry about the time because the magnetic charge is conserved. We will discuss this shortly. As we mentioned earlier, at every point of the manifold $\mathbb{R}^3$$-$$\{0\}$ (which is homotopy equivalent to  a $\mathbf{S}^2$ that surrounds the origin), we have $\mathrm{d}\mathcal{F}=0$. This $\mathbf{S}^2$ can be seen as two hemispheres glued together (in an overlap region) such that the resulting sphere  surrounds the origin. So we can define: $ \mathcal{F}=\mathrm{d}A$, locally. In the overlap region, the connections of each of these coordinate neighbourhoods are related to each other by a gauge transformation. For the case of the line, first suppose that the line is placed at the $z$-axis, therefore, it is shrinkable to a point and localized at $x_2=x_3=0$. So on the manifold $\mathbb{R}^3$$-$$\{z$-$\mathrm{axis}\}$, we have $\mathrm{d}\mathcal{F}=0$. But this time the two hemispheres should be infinitely large in order to cover the infinite line. Therefore, we can still conclude that $\mathbb{R}^3$$-$$\{z$-$\mathrm{axis}\}$ is homotopy equivalent to  $\mathbf{S}^2$, but this time it should be an imaginary infinite sphere on which we can locally define $\mathcal{F}=\mathrm{d}A$. This is in agreement with the relation between the redefined and the usual field strengths  $\mathcal{F}$ and $F$, which, as we mentioned earlier, can be obtained  by focusing only on the D-term of the full theory:
\begin{align}\label{4}
&\mathcal{F}_{23} \equiv F_{23}+ f(\Phi) \; \delta^{(2)}(x_2,x_3), \\ \nonumber
&\mathcal{F}_{\mu \nu} \equiv F_{\mu\nu} \;\;\; \text{for} \;\;\;\mu\nu\neq 23,
\end{align} 
where $f(\Phi)$ is a real function of a complex scalar field $\Phi$, which belongs to a chiral multiplet $\boldsymbol{\Phi}$, contributes to the D-term and is localized at the line. In the absence of the line, the equation of motion (\ref{3}) is invariant under a $4d$ Lorentz transformation; namely, an element of the group SO(1,3). All in all, it turns out that our set-up contains a geometrical phase.

According to our discussions in \cite{Ref2}, the phase $ \mathrm{exp}(\frac{i}{ \hbar} S_{FI})$ (where $S_{FI}$ is the Fayet-Iliopoulos (FI) D-term which contributes to the intersection part of the full theory),  has something to do with the geometric phase. This is because we must have: $\frac{S_{FI}}{\hbar}=2\pi n$. Recall that
\begin{align}\label{5}
\frac{S_{FI}}{\hbar}=\frac{\chi}{2 \hbar} \int \mathrm{d}x^0 \; \mathrm{d}x^1 \; D(x^0,x^1,0,0),
\end{align} 
where $\chi$ denotes the FI parameter and D denotes the real non-propagating auxiliary scalar field in the vector multiplet. In the QFT of this model, the dimension of $\chi$ is $[\sqrt{\hbar}]$, but classically, we can get rid of this factor beautifully. In this limit, we can verify the dimension of the one-form gauge connection $\mathbf{A}$, either directly from the classical EoM (\ref{3}), or more conveniently from a suitable Wilson loop in the presence of an electric probe charge $e$:
\begin{align}\label{6}
\mathrm{exp} \; \big(\frac{i e}{\hbar}\oint \mathrm{d}\mathbf{r}. \mathbf{A} \big),
\end{align} 
where the closed loop is placed in the plane on which we have $x^0=x^1=\mathrm{cte}$ (which we call $\mathbb{C}$-plane), so the contributing components of $\mathbf{A}$ are only $A_{2,3}$. We obtain $[\mathbf{A}]=[\mathrm{Wb}/\mathrm{L}]$. Since $D \sim \frac{\mathrm{d}}{\mathrm{d}x} A$, we have $[D]=[\mathrm{Wb}/\mathrm{L^2}]$. Therefore, the integral in the equation (\ref{5}) has the dimension of the monopole strength $[g]=[\mathrm{Wb}]$, so we have $[\chi]=[Q]$. It looks nice but it's not enough yet. To see more of the dimensional analysis, see appendix A.

On the $\mathbb{C}$-plane, we can use the complex coordinates $z=\frac{1}{2}(x_2+ \mathrm{i}x_3)$.

In appendix B, we proved that on an infinitesimal loop in the $\mathbb{C}$-plane which surrounds the line, $A_z \propto 1/z$ is always a valid nonzero solution for classical EoM.

We can continue this solution to the regions far away from the line. We put these two together in a single real equation
\begin{align}\label{9}
\bar{\partial} A_z+c.c. \propto 4 \pi \delta^{(2)} (z,\bar{z}),
\end{align} 
which can be regarded as the BPS equation for $A_z$ which holds all over the $\mathbb{C}$-plane. Moreover, it is a FODE as expected.\footnote{This construction enables us to reproduce the D-term, independently. According to (\ref{9}), we can define $D=\bar{\partial} A_z+c.c. -4 \pi \delta^{(2)} (z,\bar{z})$. So the BPS equation for $D$, should be: $D|_{\mathrm{out \; of \; the\;  line}}=0$, and $D|_{ \mathrm{on \; the \; line}}$, should not contain a delta function. Note that the gauge transformations associated with those components tangent to the defect, do not affect the D-term, so we are in the WZ gauge as expected.}

Therefore, $D(x^0,x^1,0,0)$ remains fixed. According to this proof, we believe that nothing can keep us away from  regarding the relation (\ref{5}) as the dimensionless geometrical phase,  i.e. $S_{FI}/\hbar=2 \pi n$ where $[\chi]=[Q]$.

But the best thing is yet to come. In an overlap region, the general form of the geometrical phase for either the case of flux  (in a complex line bundle) or monopole (in a principal bundle) is
\begin{align}\label{10}
\frac{e}{\hbar} \Phi_m \;\;\;\; \mathrm{or} \;\;\;\;\frac{e}{\hbar} g=2 \pi n,
\end{align}
where $\Phi_m$ is  magnetic flux as in the equation (\ref{1}) and $g$ is the strength of a monopole in the Weber convention. As we explained below the equation (\ref{2}), in physical measurements, it does not yield the correct answer for the magnetic quantity. The integral which appears in the equation (\ref{5}) represents the homotopy class to which this bundle belongs, so without any extra coefficient, it should be equal to  $g$. Therefore, it satisfies
\begin{align}\label{11}
-2 \pi \int_{\mathbf{S}^2} c_1(\mathcal{F}) = \int_{\mathbf{S}^2} \mathcal{F}=\int_{\mathrm{on \;  the \; line}} \mathrm{d}^2x \; D=g,
\end{align}
which, as we mentioned in \cite{Ref2}, can be regarded as SYM statement of the Stokes' theorem. In the equation (\ref{11}), $c_1(\mathcal{F})$ denotes the first Chern class of the principal bundle.\footnote{In mathematics literature, you may find it as $c_1(\mathcal{F})=\frac{i \mathcal{F}}{2\pi}$. The relation between the two is given by $\mathcal{F}_{\mathrm{math}}=i\mathcal{F}_{\mathrm{Maxwell}}$.\label{footnote_1}}

Now, the equation
\begin{align}\label{12}
\frac{e}{\hbar} g=\frac{\chi}{2 \hbar} g=2 \pi n
\end{align}
yields $\chi=2e$, the fundamental unit of the electric charge appearing in the measurement of the phase for a fermionic test particle (object). For $e=1$, FI parameter equals to the Euler characteristic of the 2-sphere.

In equation (\ref{11}), we can let $x^0$ transform under a Wick rotation, then we should use $\mathcal{F}_{\mathrm{math}}$ from the footnote [\ref{footnote_1}]. In the dual theory, $g$ denotes  the electric charge $e$ and $\chi/2$ is the quantized magnetic charge.

Now, we are in the position to generalize our results. Suppose that in a 4d supersymmetric field theory with the compact gauge group U(1), one can find a 2d surface $M$ equipped with a closed two-form $\Omega$, with a nonzero $c_1(\Omega)$. Therefore, we need at least two coordinate neighbourhoods to cover $M$. The phase in the overlap region is determined by ($\hbar=e=1$)
\begin{align}\label{13}
-\frac{\chi}{2} \int_M c_1(\Omega)=n.
\end{align}

In string theory, since we have two charged matter multiplets on the line, there are two independent copies of this argument.

To give further evidence, let us finish this section by considering the case of $N$ charged chiral multiplets $\boldsymbol{\Phi}_i$. The values of the complex scalar fields $\Phi_i$ on the boundaries of the line determine the total magnetic charge $Q_m^ {\mathrm{tot}}$. We denote the boundary values of $\Phi_i$ by $\varphi_i$, so we have:
\begin{align}\label{25}
Q_m^{\mathrm{tot}} &=Q_m^1+ \cdots+ Q_m^N = \frac{1}{2} \big(|\varphi_1|^2+ \cdots + |\varphi_N|^2 \big) \\ \nonumber
 &=2 \pi m_1+ \cdots + 2 \pi m_N \equiv 2 \pi n,
\end{align}
where $m_i \in \mathbb{Z}$. Note that all sources belong to the same bundle which is represented by $2\pi n$.\footnote{It is compatible with the splitting of $\chi$ as follows: $\chi=\sum_{i=1}^N \chi_i$ } For each term in (\ref{25}), the FI (large) gauge transformation corresponds to a shift of $\chi_i/2=2 \pi m_i$. Therefore, in the dual theory we have $\frac{\chi}{2}=g=2 \pi n$, which implies that in the original theory in the SI units we must have $\frac{\chi}{2}=e$.
 
In order to find the BPS equations of the monopole in the Euclidean signature in which we have $(x_E^0,x_E^1) \subset \mathbb{R}^2$ , we should choose a suitable  K\"ahler manifold as the target manifold $M$ of the sigma model, on which the local complex coordinate $z$ is defined by the map $z: \mathbb{R}^2 \rightarrow M$. According to  (\ref{25}), in order to get a correct map, the convenient non-trivial compact manifold should be obtained by performing a one-point compactification on $\mathbb{R}^2$, namely, $\mathbb{R}^2 \cup \{\infty\}     =  \mathbf{S}^2$. Therefore, the target manifold admits the standard Fubini-Study metric on $\mathbb{C}P^1$. To check the periodicity of the holonomy, one can easily use the stereographic projection to read off $\mathrm{Im}(\mathrm{ln}\;  z)=\phi$, where $z$ is the complex coordinate of $\mathbb{C}P^1$, and $\phi$ is the $2\pi$ periodic azimuthal angle of $\mathbf{S}^2$.

One set of linearly independent solutions for $D$ in relation (\ref{11}) which satisfy $g=\int D= 2\pi n$, can be defined with the aid of the Gaussian integral by
\begin{align}
D=2 \; \mathrm{exp} \Big( -\frac{(x_E^0)^2+(x_E^1)^2}{n} \Big),
\end{align}
where $n \in \mathbb{Z}_{>0}$. One can read off the $\phi$ angle inside the integral by switching to polar coordinates.

In contrast to MPS \cite{Ref4}, we should emphasize that since the target is a K\"ahler manifold with non-zero first Chern class, one can prove that it does not admit a Ricci flat metric.

Finally, we finish this section by introducing the {\it Wick rotation of the holonomy}, as a 1d map between Lorentzian  and Euclidean  holonomy:
\begin{align}
\frac{1}{2} | \Phi |^2   \longrightarrow   \mathrm{Im}(\mathrm{ln}\;  z).
\end{align}

\section{The case of $\mathcal{N}=1$, supergravity}

The case of $\mathcal{N}=1$ supergravity in the presence of the FI term has been studied in \cite{Ref7,Ref8}. Now, we want to investigate it in terms of the arguments presented in the previous section.

In \cite{Ref9}, Bagger and Witten found that the $\mathcal{N}=1$ supergravity can be coupled to a NLSM. Despite the fact that in super Yang-Mills theories (i.e. the case of no gravity) the K\"ahler transformation:
\begin{align}\label{14}
K \longrightarrow K+f+f^*,
\end{align}
(where $K$ is the  K\"ahler potential and $f$ is an arbitrary holomorphic function) is a symmetry, in  supergravity, this transformation is accompanied by some phase factors. Hence, it is a transformation between at least two local coordinate neighbourhoods on a non-trivial target space. Therefore, the target space $M$ is a K\"ahler manifold (i.e. there exists a closed two-form $\Omega$ on $M$ which is called  the K\"ahler form) with nonzero first Chern class. Without loss of generality, we can only consider the case in which the (real) dimension of $M$ is two. This is called the Bagger-Witten line bundle which indeed is a canonical line bundle over $\mathbb{C}P^1$.

According to our discussions, in the presence of the FI-term, we are allowed to use the relation (\ref{13}) (see table \ref{26}).
 \begin{table}
\setlength\belowcaptionskip{-10pt}
 \begin{center}
   \scalebox{1}{
     \begin{tabular}{|c|c|c|}
      \hline
      $\;$ & \thead{supersymmetric   \\ Dirac  monopole}  & \thead{supergravity \\ coupled to NLSM}\\ \hline
   %    \thead{$M$} &  \thead{$\mathbf{S}^2$} & \thead{$\mathbb{C}P^1$} \\ \hline
      \thead{fiber bundle \\ structure} &  \thead{U(1) bundle \\ over $\mathbf{S}^2$} & \thead{canonical line bundle  \\ over $\mathbb{C}P^1$} \\ \hline
      \thead{closed 2-form}  &   \thead{field strength  \\ form $\mathcal{F}$}   &   \thead{K\"ahler form $\Omega$}  \\ \hline
      \thead{the nonzero  \\  first Chern class}  & \thead{determines the \\ U(1) flux, $g$} & \thead{determines the  \\ parameter $\chi$} \\ \hline
      \thead{ transformation in the  \\overlap region}  &   \thead{gauge \\ transformation} &  \thead{K\"ahler transformation}  \\ \hline
      \end{tabular}
      }
      \caption{The relation between the supersymmetric Dirac monopole and $\mathcal{N}=1$, supergravity coupled to a NLSM.}
      \label{26}
  \end{center}
   \label{kt1}
\end{table}
Note that the only sphere which admits a complex structure is $\mathbf{S}^2$, and since $\mathbf{S}^2 \simeq \mathbb{C}P^1$, it is obvious that $\mathbf{S}^2$ is a K\"ahler manifold.

In supergravity, $\chi$ does not relate to the space-time U(1) flux directly. However, one can still determine it as a free parameter of the theory in the (semi)classical limit. In order to use (\ref{13}), we should calculate the first Chern class of the line bundle which can be computed by using the Fubini-Study metric.

Let $L\xrightarrow{\; \pi \;} \mathbb{C}P^1$ be the canonical line bundle over $\mathbb{C}P^1$. By using the Fubini-Study metric, we realize that the closed curvature (1,1)-form is
\begin{align}\label{15}
\Omega=i \partial \bar{\partial}\;  \mathrm{ln} (1+|z|^2)=i \frac{\mathrm{d}z \wedge \mathrm{d}\bar{z}}{(1+|z|^2)^2}.
\end{align}
Moving to the real coordinates, one can obtain
\begin{align}\label{16}
\Omega=2 \frac{\mathrm{d}x \wedge \mathrm{d}{y}}{(1+x^2+y^2)^2}=2 \frac{r \mathrm{d}r \wedge \mathrm{d}{\theta}}{(1+r^2)^2}.
\end{align}
From the relation $c_1({\Omega})=-\frac{1}{2\pi} \Omega$, we have
\begin{align}\label{17}
c_1(\Omega)=- \frac{1}{\pi}  \frac{r \mathrm{d}r \wedge \mathrm{d}{\theta}}{(1+r^2)^2}.
\end{align}
We denote the integral of $c_1(\Omega)$ over $\mathbf{S}^2$ by $C_1(L)$, which is an integer:
\begin{align}\label{18}
C_1(L)=- \frac{1}{\pi} \int \frac{r \; \mathrm{d}r \mathrm{d}\theta}{(1+r^2)^2}=- \int_1 ^{\infty} t^{-2} \mathrm{d}t=-1.
\end{align}
Plugging this into (\ref{13}) yields the final result:
\begin{align}\label{19}
\chi=2n,
\end{align}
which is consistent with the evaluation of $\chi$, derived in \cite{Ref7,Ref8}. Note that we expect the phase to appear in the fermionic matter field transformation, which is also true.

Newton's constant $G$ appears in the relation between the 4d FI parameter as a surface charge density and the 2d FI parameter, as follows: $\chi_{4d}=l_p^{-2 }\; \chi_{2d}=\frac{c^3}{\hbar G} \; (2e)$.

\section{Conclusion}

Superconductivity is a condense matter analog of the Higgs mechanism. A condensate of Cooper pairs caused a spontaneous breaking of the U(1) gauge symmetry. One can experimentally measure a magnetic flux $\Phi_m$ externally applied to a superconducting ring. By using the experimental value of the quantum of this flux $\Phi_0$, a fundamental electric charge $q$ can be determined as follows: $q=2 \pi \hbar / \Phi_0$. The charge was found to have a value of $q=2e$. This is in agreement with the fact that so far, there is no other experimental method for measuring the quantum of magnetic flux except that by using a superconductor for trapping it (thanks to the Meissner effect). This fundamental electric charge then coincides with an effective charge of two electrons of an effective quasiparticle, in other words, a Cooper pair. It has been claimed that {\it any} measurement of magnetic flux e.g. magnetic flux of a supersymmetric Dirac monopole, also reveals this fundamental electric charge. Therefore, the conservation of charge implies that in a classical scattering process, we are allowed to anticipate a decay of a magnetic monopole to a (or integer number of) pair(s) of leptons with the same electric charge. For early universe monopoles, such a decay can also address the baryon asymmetry problem.

However, evaluation of such a peculiar vertex in the conventional QFT is indeed a formidable task. Anyhow, a {\it universal interaction between electron 
and magnetic field} means that at the level of QM, the g-factor for a moving electron {\it inside} a magnetic region is equivalent to the Cooper factor for moving a electron {\it outside} a magnetic (or singular source) region.

In a conventional quantum U(1) gauge theory with a 4d Minkowski space-time as the base space, the bundle $P=\mathbb{R}^{(1,3)} \times \mathrm{U(1)}$ is a trivial one with a single local trivialization over base space. In the case of a non-trivial base space, we can use an old technique to build a soft monopole from a singular one: we know that an ordinary Dirac monopole can be submerged in a SU(2) gauge theory. Then, by introducing a scalar field $\phi$ with a constant VEV and performing a special SU(2) gauge transformation, the Dirac string disappears and we end up with the usual hedgehog form of the gauge connection of the `t Hooft-Polyakov monopole. Ultimately, the monopole source is replaced by the Higgs field. One can show that $g_{\mathrm{`tP}}=2g_{\mathrm{D}}$. Therefore, replacing $g_\mathrm{D}$ by $g_{\mathrm{`tP}}$ in  (\ref{10}) gives rise to the standard form of the condition that yields the correct result (see below the equation (\ref{2})).

A defect in 4d U(1) gauge theory generates a violation of Lorentz invariance. Therefore, the Lorentz gauge is not reliable in this case. Instead, one can use the solution (\ref{7}) for eliminating one of the DoFs.

\hfill

%\appendix*       appendix bedoone shomare
\appendix
\section{More on the dimensional analysis}\label{appen1}

The free Maxwell action in vacuum with the dimension of $\hbar$ is given by
\begin{align}\label{aa1}
S=-\frac{1}{4 \mu_0}  \int \mathrm{d}t \; \mathrm{d}^3 x \; \mathcal{F}^{\mu\nu}\mathcal{F}_{\mu\nu}=-\frac{1}{4 \mu_0 c} \int \mathrm{d}^4 x \; \mathcal{F}^{\mu\nu}\mathcal{F}_{\mu\nu}.
\end{align}
So one might be tempted to insert this $\mu_0 c$ factor into the problem. Therefore, the Wilson loop becomes
\begin{align}\label{aa2}
\mathrm{exp} \; \Big(\frac{i \chi}{2\hbar \mu_0 c}\int \mathrm{d}^2x \; D \Big)=\mathrm{exp} \; \Big(\frac{i \chi}{2\hbar}\oint \mathrm{d}\mathbf{r}. \frac{\mathbf{A}}{\mu_0 c} \Big).
\end{align}\label{aa3}
In principle, we can use the duality transformation:
\begin{align}
\frac{1}{c} j_m \longrightarrow \mu_0 j_e,
\end{align}
from which we find the dimension of the dual gauge connection:
\begin{align}\label{aa4}
\mathbf{[A]}=\Big[ \frac{\mathrm{Wb}}{\mathrm{L}} \Big] =\Big[ \frac{\mu_0 c Q}{\mathrm{L}} \Big].
\end{align}
In this sense, we realize that nothing has changed, i.e.  $\chi /2$ represents $g$ with the dimension of $\mathrm{Wb}$, and the integral represents $e$, with the dimension of $Q$.

Note that the relation (\ref{aa4}) implies that a topological charge $q_m$, is associated with an electric charge $q_e$, by
\begin{align}\label{aa5}
q_m & \longrightarrow \mu_0 c q_e  \;\;\;\;\; (\mathrm{Weber  \; convention}),\\ \nonumber
q_m & \longrightarrow c q_e \;\;\;\;\; (\mathrm{Ampere} \textrm{-} \mathrm{meter \; convention}).
\end{align}

\section{A theorem}\label{appen2}
\begin{thm}\label{th1}
On an infinitesimal loop in the $\mathbb{C}$-plane which surrounds the line, $A_z \propto 1/z$ is always a valid nonzero solution.
\end{thm}
\begin{proof}
It follows from Poincar\'e's lemma for a non-trivial topology that there always exist at least two solutions of which at least one is nonzero which we denote by $\mathbf{A}$, and the other by $\boldsymbol{\mathcal{A}}$. First, we claim that for a given solution $\boldsymbol{\mathcal{A}}$ on this loop the below relation,
\begin{align}\label{7}
\mathcal{A}_z \longrightarrow \mathcal{A}_z+\frac{1}{z} \approx \frac{1}{z},
\end{align} 
is a gauge transformation. So we must show that $1/z$ is a pure gauge on the loop. This is true  because the integrand of the integral of $1/z$ over the loop, namely,
\begin{align}\label{8}
\oint \mathrm{d}{z} \frac{1}{z}=\mathrm{i} \oint \mathrm{d} \phi,
\end{align} 
is a total derivative where $\phi$ is the azimuthal angle that parametrizes the infinitesimal loop. Since at every point out of the line, the gauge transformation of any given connection $\boldsymbol{\mathcal{A}}$ gives another connection, so at every point on the infinitesimal loop we find $A_z \propto 1/z$ is also a valid nonzero solution. QED 
\end{proof}

%--------------------------------------------------------------------------

\end{document}